\documentclass[]{article}
\pdfoutput=1
\usepackage{cite}
\usepackage{amsmath}
\usepackage{amsfonts}
\usepackage{amssymb}
\usepackage{amsthm}
\usepackage{enumitem}
\usepackage{hyperref}
\usepackage{color}
\usepackage{caption}
\usepackage{tikz-cd}
\usepackage{rotating}
\usepackage[english]{babel}
\usepackage[utf8]{inputenc}
\usepackage{breqn}
\newtheorem{theorem}{Theorem}
\newtheorem{lemma}[theorem]{Lemma}
\newtheorem{cor}[theorem]{Corollary}
\newtheorem{prop}[theorem]{Proposition}
\newtheorem{conjecture}[theorem]{Conjecture}
\newtheorem{case}{Case}
\theoremstyle{definition}
\newtheorem{defn}[theorem]{Definition}
\newtheorem{example}[theorem]{Example}
\newtheorem{remark}[theorem]{Remark}
\newtheorem{question}[theorem]{Question}
\newtheorem{problem}[theorem]{Problem}
\numberwithin{theorem}{section}

\DeclareMathOperator{\tr}{tr}

\DeclareMathOperator{\Sym}{Sym}

\DeclareMathOperator{\Mat}{Mat}
\DeclareMathOperator{\Spec}{Spec}
\DeclareMathOperator{\IMPS}{uMPS}
\newcommand{\IMPSbis}{\normalfont{\textrm{uMPS}}}
\DeclareMathOperator{\Cyc}{Cyc}
\DeclareMathOperator{\diag}{diag}
\DeclareMathOperator{\Span}{span}
\DeclareMathOperator{\jn}{join}

\def\CC{\mathbb{C}}
\def\RR{\mathbb{R}}
\def\KK{\mathbb{K}}
\def\PP{\mathbb{P}}
\def\ot{\otimes}
\providecommand{\keywords}[1]{\textbf{Keywords} #1}
\providecommand{\msc}[1]{\textbf{Mathematics Subject Classification (2000)}  #1}

\title{Uniform matrix product states from an algebraic geometer's point of view}
\author{Adam Czapli\'nski, Mateusz Micha\l{}ek, Tim Seynnaeve}

\begin{document}
\selectlanguage{english}
\maketitle

\begin{abstract}
We apply methods from algebraic geometry to study uniform matrix product states. Our main results concern the topology of the locus of tensors expressed as $\IMPS$, their defining equations and identifiability. By an interplay of theorems from algebra, geometry and quantum physics we answer several questions and conjectures posed by Critch, Morton and Hackbusch.
\end{abstract}

 
\vspace*{2ex}

\msc{15A69, 15A90}

\tableofcontents

\section{Introduction}
Matrix product states and uniform matrix product states play a crucial role in quantum physics and quantum chemistry \cite{PerezGarciaVerstraete, QuantumWielandt, orus2014practical, schollwock2011density, hackbusch2012tensor, ye2018tensor}.
They are used, for instance, to compute the eigenstates of the Schr\"odinger equation. 
Matrix product states provide a way to represent special tensors in an efficient way and uniform matrix product states are partially symmetric analogs of matrix product states. 

We introduce new algebraic methods to answer problems and questions coming from this area. 
In particular, we answer several questions and conjectures posed by Critch, Morton and Hackbusch.
Our main emphasis is on interactions among algebraic geometry and matrix product states. While secant varieties and border rank already form a well-established topic within algebraic geometry \cite{landsberg2012tensors,landsberg2017geometry,iarrobino1999power, zak2005tangents}, tensor networks and matrix product states so far have only received very limited attention (with notable exceptions \cite{landsburg2012geometry, CritchMorton, PEPS}). However, we are sure that this situation is changing.

One of the main goals of our paper is to study geometric and topological properties of the uniform matrix product states $\IMPS(D,d,N)$ and the Zariski closure $\overline{\IMPS(D,d,N)}$, which in case of complex numbers coincides with the Euclidean closure.
For the precise definition, see Definition \ref{def:uMPS}.

As an application of our methods we confirm two conjectures of Critch and Morton \cite{CritchMorton}.
Precisely, one asserts that, under mild assumptions, for two $2\times 2$ matrices, matrix product states are ``identifiable'' - details are explained later. The main ingredient of our proof is the so-called \emph{fundamental theorem of uniform matrix product states}. 
We can also verify the conjecture concerning the defining ideal of the closure of the $\IMPS$ in special cases. 

Further, using representation theory we are able to provide a full description of the $\IMPS$ in the case of products of length $4$ of two $2\times 2$ matrices. So far only the defining equation of the variety was known - we provide a full description of the dense, proper subset $\IMPS(2,2,4)$.
This is related to the probabilistic graphical models known as hidden Markov models  and to the conjecture of Bray-Morton-Sturmfels \cite{BrayMortonHMM}.
A variant of this conjecture states that for any fixed $D$ and $d$, the ideal of $\IMPS(D,d,N)$ is generated by quadrics for $N$ large enough.
One of the tools we use is the \emph{trace algebra} \cite{Procesi,Sibirskii}. The applications of trace algebras to the theory of matrix product states were already investigated by Critch and Morton \cite{CritchMorton}. 
We show how this method can be used to derive the conjectured description of the ideal of $\overline{\IMPS(2,2,5)}$. Further, we provide a full description of the ideal of $\overline{\IMPS(2,2,6)}$.	

Moreover, we describe a useful surjectivity criterion for polynomial maps, which can be used to prove that every tensor can be expressed as a $\IMPS$. We apply it to show that $\IMPS(3,2,4)$ fills the ambient space of cyclic tensors. Our method is very general and we believe it can be used in many other cases beyond matrix product states.

The very important questions of the closedness of families of tensors that allow
representations as matrix product states were asked by W. Hackbusch and L.~Grasedyck (cf.~\cite{landsburg2012geometry, HarrisMichalekSertoz}).
One of the questions was, when $\IMPS(D,d,N)$ and $\overline{\IMPS(D,d,N)}$ may differ. We answer the question when $\IMPS(D,d,N)$ is closed, i.e.~both sets are equal, in the case $D=2$. Further, we provide an explicit tensor (the so-called \emph{$W$-state}) which is always in $\overline{\IMPS(D,d,N)}$, but not in $\IMPS(D,d,N)$ when $N$ is large compared to $D$, providing many instances where $\IMPS(D,d,N)$ is not closed. Making this precise requires investigating the so-called injectivity radius.
Moreover, we study the dimension of $\overline{\IMPS(D,d,N)}$ and the connectedness in a more general set-up.
In the table below we collect our results concerning closedness of $\IMPS(D,2,N)$. The second row is Theorem \ref{ClosedComplex}, the red F is Example \ref{eg:324fills}.
 
 \newcommand{\diago}{%
 	\setlength{\unitlength}{1cm}
 	\begin{picture}(0,0)(.35,-.25)
 	\put(0,0){\line(3,-2){.6}}
 	\put(0,-0.35){\footnotesize D}
 	\put(0.35,-0.2){\footnotesize N}
 	\end{picture}}

 \newcommand{\mc}[1]{\multicolumn{1}{c}{#1}}
 \newcommand{\mytab}{
 	 \begin{tabular}{cc|c|c|c|c|c|c|c|} 
 		&\mc{$x$} &\mc{2} &\mc{3} &\mc{4} &\mc{6} &\mc{8} &\mc{14} &\mc{20} \\
 		$y$&\diago &  1    &  2    &  3    &  4    &     5 &  6     &     7  \\ \cline{2-9}
 		2    &   1  &  F    &  C    &  C    &  C    &     C &  C     &     C  \\ \cline{2-9}
 		5    &   2  &  F    &  F    &  F    &    N   &  N     &  N      &   N     \\ \cline{2-9}
 		10   &   3  &  F    &  F    &  F    &   {\color{red}F}    &       &        &        \\ \cline{2-9}
 		27   &   4  &  F    &  F    &  F    &   F    &       &        &        \\ \cline{2-9}
 	\end{tabular}
 }

\begin{table}[h!] 
	\centering
	\mytab
	\captionsetup{justification=centering}
	\caption{$\IMPS(D,2,N)$ \\
		F: fills the ambient space,
		C: closed, but does not fill,
		N: not closed \\
		$x$: dimension of the ambient space,
		$y$: expected dimension } \label{thetable}
\end{table}

Above questions are often inspired by applications. However, they may be also seen as analogues of well-studied questions in the theory of secant varieties, rank and border rank. For example, the closedness of the image of a map is a natural question from a theoretical point of view, but also plays an important role in best approximation problems \cite{de2008tensor, qi2017complex}.

\subsubsection*{Plan of the article.}
Let us describe the content and the structure of this paper.
For more detailed
explanations we refer to the introductions of the individual sections.
Our article 
contains an introduction, two main sections 
and two appendices. In Section \ref{section:def}  we collect basic definitions, results and notations regarding
uniform matrix product states $\IMPS(D,d,N)$. We discuss the closedness in trivial cases, dimension, local symmetries 
and we define the so-called $W$-state.  

The next Section \ref{section:mainresults} consists of four parts.
In Section \ref{section:topprop} we give a complete classification when $\IMPS(2,d,N)$ is closed. Then we discuss the closedness in other cases.
We also introduce the \textit{injectivity radius} $C_D$ and we prove connectedness of the $\IMPS$.
In Section~\ref{section:surjectivity} we explore for which parameters the set $\IMPS(D,d,N)$ fills the ambient space $\Cyc^N(\mathbb{C}^d)$.
In Section \ref{section:tracepar} we recall another parametrization of the matrix product states.
Using this \textit{trace parametrization} and \textit{Macaulay2} \cite{M2} we obtain defining equations
for $\IMPS(2,2,N)$ for small values of $N$.
The last Section \ref{section:fundtheorem} is devoted to the new results related to the fundamental theorem of matrix 
product states. At the end of our article we have placed two appendices.
In Appendix~A we collect the technical part of the proof of the Theorem~\ref{ClosedComplex}.
In Appendix~B we give a full
description of $\IMPS(2,2,4)$ as a constructible subset of $\Cyc^4(\CC^2)$. 

\section{Definitions and basic properties} \label{section:def}

In this section we introduce basic notions and theorems.
We work over $\mathbb{C}$, unless explicitly stated otherwise.
We start by recalling the definition of a uniform matrix product state:
\begin{defn} \label{def:uMPS}
	We fix 3 parameters $D,d,N \in \mathbb{N} \setminus \{0\}$.
	A tensor $T \in (\mathbb{C}^d)^{\otimes N}$ is called a \emph{uniform $D$-matrix product state} if there is a collection of $d$ matrices $M_0,\ldots,M_{d-1}$ in $\mathbb{C}^{D \times D}$ such that  
	\begin{equation} \label{eq:uMPS}
	T = T_N(M_0,\ldots,M_{d-1}) := \sum_{0 \leq i_1, \ldots, i_N \leq d-1}{\tr(M_{i_1} \cdots M_{i_N})e_{i_1}\otimes \cdots \otimes e_{i_N}} \text{.}
	\end{equation}
	We will sometimes write $T_{D,d,N}$ instead of $T_N$ if the parameters are not clear from the context.
The set of all uniform $D$-matrix product states in $(\mathbb{C}^d)^{\otimes N}$ is denoted by $\IMPS(D,d,N)$.
In other words, $\IMPS(D,d,N)$ is the image of the polynomial map
\begin{align*}
T_N:(\mathbb{C}^{D \times D})^d \to (\mathbb{C}^d)^{\otimes N} \text{.}
\end{align*}
\end{defn}
From now on we will leave the parameter $D$ implicit and simply use the terminology \emph{uniform matrix product state}. We point out that every cyclically symmetric tensor will be a matrix product state for $D$ large enough (Corollary \ref{fillsforNlarge}).
\begin{remark}
	The motivation for considering matrix product states comes from quantum information theory, where they describe physically meaningful states (more precisely: ground states of local hamiltonians) of physical systems that consist of $N$ sites placed in a ring, where each site has a state space of \emph{physical dimension} $d$. The parameter $D$, called \emph{bond dimension}, indicates how strong the entanglement between neighboring sites can be. We refer the reader to \cite[Section 1]{PEPSopen} for an overview, where matrix product states are introduced as a special class of so-called \emph{tensor network states}.
	
	We briefly explain the adjective ``uniform": one can define matrix product states (MPS) more generally by allowing the matrices $M_i$ to depend on the physical site. An MPS is then a tensor of the form 
	\begin{equation} \label{eq:nonuniform}
	\sum{\tr(M^1_{i_1} \cdots M^N_{i_N})e_{i_1}\otimes \cdots \otimes e_{i_N}},
	\end{equation}
	where we now have $N$ matrix tuples $(M^j_0, \ldots, M^j_{d-1})$. In this setup one can allow the size of the matrices, as well of the number $d$ of matrices in a tuple, to depend on the tuple (as long as the products appearing in (\ref{eq:nonuniform}) make sense). An important special case is that of \emph{open boundary conditions}, where the matrices $M^1_i$ (resp.\  $M^N_i$) are row vectors (resp.\ column vectors) and hence the traces in (\ref{eq:nonuniform}) are traces of $1 \times 1$-matrices. Tensors of this form also arise in numerical optimization, under the name \emph{tensor train format} \cite{oseledets2011tensor}.
	
	Uniform matrix product states are simply matrix product state where the tuple $(M^j_0, \ldots, M^j_{d-1})$ is the same for every $j$. They are also known as \emph{translation invariant MPS} or \emph{site-independent MPS}. 
\end{remark}
The set $\IMPS(D,d,N)$ is a cone, i.e.~if  $T\in\IMPS(D,d,N)$, then also $\lambda T \in \IMPS(D,d,N)$ for every $\lambda \in \mathbb{C}$. This is no longer true if the field is not algebraically closed, e.g.\ for $\mathbb{R}$ it is only guaranteed if $N$ is odd or $\lambda \geq 0$.

Since $\tr(M_{i_1} \cdots M_{i_N})$ does not change if we cyclically permute the matrices in the product, it follows that $\IMPS(D,d,N) \subseteq \Cyc^N(\mathbb{C}^d)$, where  $\Cyc^N(\mathbb{C}^d) \subseteq (\mathbb{C}^d)^{\otimes N}$ is the subspace of cyclically symmetric tensors.

As $\IMPS(D,d,N)$ is the image of a polynomial map, it is a constructible set (i.e.\ a finite union of locally Zariski closed sets) by Chevalley's theorem \cite[IV, 1.8.4.]{EGA}. Its Euclidean closure $\overline{\IMPS(D,d,N)}$ agrees with its Zariski closure and is an algebraic variety. In Section \ref{section:tracepar} we will give defining equations for small parameter values, and in Appendix \ref{appendix:224} we give a complete description of the smallest nontrivial case $\IMPS(2,2,4)$.

A natural question to ask is the following:
\begin{question}\label{quest:closed}
	For which parameters $D,d,N$ is $\IMPS(D,d,N)$ a closed set?
\end{question}
Analogous questions have been investigated from the point of view of complex and real tensors of bounded rank. In that case, most often, the locus is not closed leading to the central notion of border rank \cite{landsberg2012tensors, landsberg2017geometry, de2008tensor, qi2017complex, seigal2017real, buczynska2014secant}. 

Question \ref{quest:closed} will be the main subject of Section \ref{section:topprop}.
Below we collect some easy results regarding closedness of $\IMPS(D,d,N)$. The next lemma follows immediately from the definitions.
\begin{lemma}
If $D \leq D'$, then $\IMPS(D,d,N) \subseteq \IMPS(D',d,N)$.
\end{lemma} 
\begin{cor} \label{inductD}
	If $\IMPS(D,d,N) = \Cyc^N(\mathbb{C}^d)$, then for any $D'\geq D$, \linebreak $\IMPS(D',d,N) = \Cyc^N(\mathbb{C}^d)$. In particular, in such a case 
	$\IMPS(D',d,N)$ is closed for all $D'\geq D$.
\end{cor}
\begin{lemma} \label{lemma:proj}
For $d \leq d'$ we have an inclusion $(\mathbb{C}^d)^{\otimes N}\subset (\mathbb{C}^{d'})^{\otimes N}$. The following equality holds: 
$$\IMPS(D,d,N) = \IMPS(D,d',N) \cap (\mathbb{C}^d)^{\otimes N}.$$ 
\end{lemma}
\begin{proof}
	The inclusion $\subseteq$ follows from the equality
	\[
	T_N(M_0,\ldots,M_{d-1}) = T_N(M_0,\ldots,M_{d-1},0,\ldots,0) \text{.}
	\]
	For the other inclusion we note that the projection $(\mathbb{C}^{d'})^{\otimes N} \to (\mathbb{C}^d)^{\otimes N}$  maps $\IMPS(D,d',N)$ to $\IMPS(D,d,N)$.
\end{proof}
\begin{cor} \label{inductd}
	If $\IMPS(D,d,N)$ is not closed, then for all $d' \geq d$, the set $\IMPS(D,d',N)$ is not closed.
\end{cor}
\begin{prop}\label{prop:TrivialCases}
	If $D=1$ or $d=1$ or $N\leq2$, then $\IMPS(D,d,N)$ is closed.
\end{prop}
\begin{proof}
	If $d=1$ or $N=1$, then by definition $\IMPS(D,d,N)=(\mathbb{C}^d)^{\otimes N}$. \\
	For $D=1$: $\IMPS(1,d,N)$ is equal to 
	the image of the Veronese embedding $\mathbb{C}^d \to \Sym^N(\mathbb{C}^d)\subset (\mathbb{C}^d)^{\otimes N}$, which is known to be closed. \\
	For $N=2$: $\IMPS(D,d,2) \subseteq (\mathbb{C}^d)^{\otimes 2} = \CC^{d \times d}$ consists of
	all symmetric $d \times d$-matrices of rank at most $D^2$,
	 	 which is a closed set.
\end{proof}

If $T\in\IMPS(D,d,N)$, there can be many different choices of matrices $M_0,\ldots,M_{d-1}$ exhibiting $T$ as $T_N(M_0,\ldots,M_{d-1})$. 
In particular, we have the following:
\begin{remark} \label{JNF}
Observe that for $P \in GL(D,\mathbb{C})$, it holds that 
	\[
	T_N(M_0,\ldots,M_{d-1}) = T_N(P^{-1}M_0P,\ldots,P^{-1}M_{d-1}P) \text{.}
	\]
In particular, 
for $T \in \IMPS(D,d,N)$, we can write $T = T_N(M_0,\ldots,M_{d-1})$ where $M_0$ is in Jordan normal form.
\end{remark}
We expect that the generic fiber of the $\IMPS$-map $(\mathbb{C}^{D \times D})^d \to \Cyc^N(\mathbb{C}^d)$ consists of $D!N$ simultaneous conjugacy classes: the $D!$ comes from permuting the rows and columns of the matrices, and the $N$ from multiplying each matrix with an $N$-th root of unity. In Section \ref{section:fundtheorem}, we will show that this is true for large $N$. It is a corollary of the \emph{fundamental theorem of uniform matrix product states}  \cite{CiracVerstraeteFundThm,Molnar}.


We now discuss the expected dimension of the $\IMPS$. 
\begin{prop} \label{expdim}
	The dimension of the variety $\overline{\IMPS(D,d,N)}$ is at most \linebreak $\min \{(d-1)D^2 + 1 , \dim(\Cyc^N(\mathbb{C}^d))\}$. 
\end{prop}
\begin{proof}
	If $T \in \IMPS(D,d,N)$, we can write $T=T_N(M_0,\ldots,M_{d-1})$. By Remark~\ref{JNF} we can, for $M_0$ and $M_1$ generic, assume that $M_0$ is diagonal and that the $D-1$ nondiagonal entries on the top row of $M_1$ are all equal to $1$. Since the dimension of the image of a polynomial map does not change if we restrict to a dense subset, we are done.
\end{proof}

\begin{remark}
	It is not hard to show that the dimension of the ambient space $\Cyc^N(\mathbb{C}^d)$ is equal to 
\[
	\dim \Cyc^N(\mathbb{C}^d) = \frac{1}{N}\sum_{\ell|N}{\varphi(\ell)d^{N/\ell}}
\]
where $\varphi$ is the Euler totient function.

	If $N$ is prime, this simplifies to
	\[
	\dim \Cyc^N(\mathbb{C}^d) = \frac{d^N+(N-1)d}{N} \text{.}
	\]
\end{remark}


\begin{defn}
	If the equality
	\[
	\dim \IMPS(D,d,N) = \min \{(d-1)D^2 + 1 , \dim(\Cyc^N(\mathbb{C}^d))\}
	\]
	holds, we say that $\IMPS(D,d,N)$ has \emph{expected dimension}.
\end{defn}
It is natural to ask for which parameters $\IMPS(D,d,N)$ has expected dimension. This question is very similar to the the problem of determining for which parameters the set of tensors of bounded tensor rank has expected dimension, which has received considerable attention in the literature. The most famous result along these lines is the celebrated Alexander-Hirschowitz Theorem \cite{alexanderHirschowitz}, which gives a complete answer in the case of symmetric tensor rank. For general tensors the question remains open --- for important partial results we refer to \cite{catalisano2011secant, landsberg2012tensors, oeding2016equations, abo2012new}. 
\begin{remark} \label{rmk:dimJac}
	The dimension of $\overline{\IMPS(D,d,N)}$ is easy to compute for small values of $(D,d,N)$ using the Jacobian criterion. In particular, we could check that for all cases in Table \ref{thetable}, $\IMPS(D,d,N)$ has expected dimension.
\end{remark}

\begin{conjecture} \label{conj:exp}
For every choice of $(D,d,N)$, $\IMPS(D,d,N)$ has expected dimension.
\end{conjecture}
We can obtain more cases for which Conjecture \ref{conj:exp} holds from our results in Section \ref{section:mainresults}. More precisely: $\IMPS(D,d,N)$ clearly has expected dimension if its closure fills the space $\Cyc^N(\CC^d)$. In Corollary \ref{fillsforNlarge}, this is shown to hold for large $D$. Additionally, $\IMPS(D,d,N)$ has expected dimension if the general fiber of the map $T_{D,d,N}$ has dimension $D^2-1$. This holds for large $N$ by the fundamental theorem (see Corollary \ref{cor:fundthm}).

\begin{remark} \label{rmk:locsym}
	The set $\IMPS(D,d,N)$ is invariant under \emph{local symmetries}. Explicitely, note that the action of $GL_d$ on $\CC^d$ induces an action of $GL_d$ on $(\CC^d)^{\otimes N}$ by
	$A\cdot (v_1\otimes\ldots \otimes v_{N})=(A\cdot v_1)\otimes\ldots\otimes (A\cdot v_N)$.
	This action restricts to the space $\Cyc^N(\CC^d)$. It is compatible with the $\IMPS$ map in the sense
	that for a matrix $A=(a_{i,j})_{ij}\in  GL_d$ and $(M_0,\ldots, M_{d-1})\in(\CC^{D\times D})^d$, it holds
	that $$T_N(\sum_j a_{0,j}M_j,\ldots, \sum_j a_{{d-1},j}M_j)=A\cdot T_N(M_0,\ldots, M_{d-1}).$$
	This implies that $\IMPS(D,d,N)$ is invariant under the action defined above. 
\end{remark}

We conclude this section by introducing some notation. \newline Let us denote $[d]=\{0,\ldots,d-1\}$.
Then for $(i_1,\ldots,i_N) \in [d]^N$, we define 
\[
e_{i_1i_2\cdots i_N} := \sum_{(j_1,\ldots,j_N)}{e_{j_1} \otimes \cdots \otimes e_{j_N}} \in \Cyc^N(\mathbb{C}^d) \subseteq (\mathbb{C}^d)^{\otimes N},
\]
where the sum is over all disctinct cyclic permutations of $(i_1,\ldots,i_N)$. For example: $e_{111}=e_1\ot e_1 \ot e_1$, $e_{110}=e_1\ot e_1 \ot e_0+e_1\ot e_0 \ot e_1+e_0\ot e_1 \ot e_1$, and $e_{1010}=e_1\ot e_0 \ot e_1 \ot e_0 + e_0 \ot e_1 \ot e_0 \ot e_1$.


\begin{defn}\label{Wstatedef}
The \emph{$W$-state} is defined as
\[
W_N := e_{0\cdots01} \in \Cyc^N(\mathbb{C}^2) \text{.}
\]
\end{defn}

The W-state plays, for example, an important role in Quantum Information Theory \cite{DVC2000}. It will also be essential in our discussion about closedness in the next section.

\section{Main results}\label{section:mainresults}

In this section we present our main results.
\subsection{Topological properties}\label{section:topprop}
We start this section by giving a complete classification when $\IMPS(2,d,N)$ is closed.
\begin{theorem} \label{ClosedComplex}
	For  $d>1$ and $N>2$, $\IMPS(2,d,N)$ is closed if and only if $(d,N) = (2,3)$. 
\end{theorem}
\begin{proof}
	The set $\IMPS(2,2,3)$ is closed because it equals the ambient space $\Cyc^3(\CC^2)=\Sym^3(\CC^2)$.
	This was already proven in \cite[Section 5.2]{HarrisMichalekSertoz}, so we only sketch the proof here: by Remark \ref{rmk:locsym}, it suffices to find one tensor in every $GL_2$-orbit that is in $\IMPS(2,2,3)$. Since there are only three $GL_2$-orbits, this can be done explicitly.
	
	Next, we show 
	that $\IMPS(2,3,3)$ is not closed. 
	To do this, it suffices to construct a tensor $T$, which is in $\overline{\IMPS(2,3,3)}$,
	but not in $\IMPS(2,3,3)$. We will prove that the following
	tensor works:
	 \[
	 T =e_{012}=e_0\otimes e_1 \otimes e_2 + e_1\otimes e_2 \otimes e_0 + e_2\otimes e_0 \otimes e_1 \in \Cyc^3(\CC^3).
	 \]
	 
	 First we show that $T \in \overline{\IMPS(2,3,3)}$ by exhibiting $T$ as a limit of tensors in $\IMPS(2,3,3)$. 
	 Consider for every $\lambda \in {\mathbb{C}\setminus\{0\}}$ the following three matrices:
	\[
	M_{0,\lambda} = \lambda^2 \begin{pmatrix} 1 & 0 \\ 0 & 0 \end{pmatrix}, \,
	M_{1,\lambda} = \lambda^{-1} \begin{pmatrix} 0 & 1 \\ 0 & 0 \end{pmatrix}, \,
	M_{2,\lambda} = \lambda^{-1} \begin{pmatrix} 0 & 0 \\ 1 & 0 \end{pmatrix} \text{.}
	\]
	Then by definition 
	\[
	T_3(M_{0,\lambda},M_{1,\lambda},M_{2,\lambda}) = \lambda^6e_{000} + e_{012},
	\]
	so we have $T = \lim_{\lambda \to 0}{T_3(M_{0,\lambda},M_{1,\lambda},M_{2,\lambda})}$, hence  $T \in \overline{\IMPS(2,3,3)}$.
Proving that $T \notin \IMPS(2,3,3)$ amounts to showing that a certain system of polynomial equations has no solutions. We proved this via a Gr\"{o}bner basis computation in \textit{Macaulay2}.
For details we refer to the appendix (Lemma \ref{appendix233}).

Next, we need to show that
 $\IMPS(2,2,N)$ is not closed for any $N>3$. 
We can exhibit the \emph{W-state} $W_N$ from Definition \ref{Wstatedef} as a limit of tensors in $\IMPS(2,2,N)$. 
Let $\zeta$ be any complex number satisfying $\zeta^N=-1$.
	For every $\lambda \neq 0$, consider the matrices 
	$ M_{0,\lambda} = \lambda^{-1}
	\begin{pmatrix}
	1 & 0 \\
	0 & \zeta
	\end{pmatrix}
	$
	and 
	$ M_{1,\lambda} = \lambda^{N-1}
	\begin{pmatrix}
	1 & 0 \\
	0 & -\zeta
	\end{pmatrix}
	$.\\
	Then it is easy to see that 
	\[
	T_N(M_{0,\lambda},M_{1,\lambda})= 2W_N + O(\lambda) \text{,}
	\]
	where $O(\lambda)$ means higher order terms in $\lambda$. 
	Letting $\lambda \to 0$ shows that $W_N \in \overline{\IMPS(2,2,N)}$.
Showing that $W_N \notin \IMPS(2,2,N)$ is done by an explicit computational argument.
We refer to Lemma \ref{appendix22N} for this.

For larger $d$, the result follows from Corollary \ref{inductd}.
\end{proof}
\begin{remark}
	The statement of Theorem \ref{ClosedComplex} is also true when working over $\RR$ instead of $\CC$. The only additional things we need to show for this are
	\begin{itemize}
		\item Also over $\RR$, it holds that $\IMPS(D,d,N)=\Sym^3(\RR^2)$. The proof is analogous to the complex case (this time there are four $GL_2$-orbits).
		\item The $W$-state is also in the closure of real-valued $\IMPS(D,d,N)$. This can be achieved by replacing the complex matrices $M_{0,\lambda}$, $M_{1,\lambda}$ in the proof by the real matrices $ M_{0,\lambda} = \lambda^{-1}
		\begin{pmatrix}
		\zeta + \zeta^{-1} & 1 \\
		-1 & 0
		\end{pmatrix}
		$
		and 
		$ M_{1,\lambda} = \lambda^{N-1}
		\begin{pmatrix}
		1 & 0 \\
		0 & 1
		\end{pmatrix}
		$.\\
	\end{itemize}
\end{remark}
We proceed to show that for any fixed $D,d$, $\IMPS(D,d,N)$ will not be closed for large $N$. This is particularly important in quantum physics, where $N$ is typically assumed to be very large.
 Before we can state the precise result, we need to introduce the \emph{injectivity radius}. For $L \subseteq \CC^{D \times D}$ a linear space of $D \times D$ matrices, $L^k$ will denote the linear space spanned by all products of $k$ matrices in $L$.
\begin{defn}
For $D \in \mathbb{N}$, the \emph{injectivity radius} $C_{D}$ is the smallest natural number such that the following holds:\\
	Whenever $L$ is a linear space of $D \times D$-matrices, either $\dim L^{C_D} = D^2$, or $\dim L^{N} < D^2$ for every $N$. 
\end{defn} 
The existence of this constant was first established in \cite{QuantumWielandt}, where it was shown that $C_D \leq D^2(D^2+1)$. Recently, it has been shown that $C_{D}=O(D^2 \log D)$ \cite{michalek2018quantum}. It is conjectured that $C_{D}=O(D^2)$.
In the case $D=2$, one can show that the injectivity radius equals three.
\begin{remark} \label{rmk:geninin}
We define the \emph{generic injectivity index} as the smallest integer $GC_{D,d}$ such that for any \emph{generic} $d$-dimensional space $L$ of $D\times D$ matrices we have $\dim L^{GC_{D,d}}=D^2$. This quantity is much easier to control, as to bound it from above it is enough to exhibit one pair of matrices that generate the whole space fast. One can prove that $GC_{D,d}$ is $O(D)$ and we conjecture that $GC_{D,d}$ is $O(\log(D))$.
\end{remark}

We are interested in the injectivity radius because of the following theorem.
\begin{theorem}[{See \cite[Corollary 1]{PerezGarciaVerstraete}}] \label{thm:verstraeteW}
	When $N > 6(D-1)(C_D+1)$, we have that $W_N \notin \IMPS(D,2,N)$.
\end{theorem}
\begin{cor}
	If $N > 6(D-1)(C_D+1)$, then $\IMPS(D,d,N)$ is not closed.
\end{cor}
\begin{proof}
	It suffices to consider the case $d=2$.
	Our arguments from the
proof of Theorem \ref{ClosedComplex} show that $W_N \in \overline{\IMPS(D,2,N)}$ for every $D$. We conclude by Theorem \ref{thm:verstraeteW}
\end{proof}
\begin{conjecture}
	Except for the trivial cases in Proposition \ref{prop:TrivialCases}, $\IMPS(D,d,N)$ is only closed if it fills the ambient space $\Cyc^N(\CC^d)$.
\end{conjecture}
We now discuss the connectedness. The set $\IMPS(D,d,N)$ is clearly connected, since every tensor in it can be rescaled to $0$. A more interesting question is whether its projectivisation is connected.
\begin{lemma} \label{lem:conn}
Let $F\colon\mathbb{P}^n\dasharrow\mathbb{P}^N$ be a rational map. Then we have the following:
\begin{enumerate}
\item over $\CC$ the image is always connected.
\item over $\RR$ the image is connected if the codimension of the base locus is greater than one.
\end{enumerate}
\end{lemma}
\begin{proof}
Denote by $Z \subseteq \mathbb{P}^n$ the locus where this map is not defined. 
We observe that $Z$ is also a subvariety of $\mathbb{P}^n$.
Now it suffices to show that the complement of $Z$ is connected. In the case $\KK=\CC$ this is immediate, in the case $\KK=\RR$ it follows from the assumption. 
\end{proof}
\begin{theorem}
	Both over $\CC$ and over $\RR$, for any choice of the parameters, the projectivisation of $\IMPS(D,d,N)$ is connected.
\end{theorem}
\begin{proof}
	The projectivisation of $\IMPS(D,d,N)$ is the image of the rational map
	\begin{align*}
	T_N:\PP((\mathbb{K}^{D \times D})^d) \dashrightarrow \PP((\mathbb{K}^d)^{\otimes N}) \text{.}
	\end{align*}
	The base locus $Z \subseteq \PP((\mathbb{K}^{D \times D})^d)$ consists af all $(M_1,\ldots,M_d)$ such that for every choice of indices $\tr(M_{i_1}\cdots M_{i_N})=0$. From this description it is clear that $Z$ is a subvariety of $\PP((\mathbb{K}^{D \times D})^d)$.\\
	By Lemma \ref{lem:conn}, we only need to check that $Z$ is not a hypersurface. But this is obvious: if $Z$ were a hypersurface, all polynomial equations coming from $\tr(M_{i_1}\cdots M_{i_N})=0$ would be divisible by the same polynomial. This is clearly not the case: the equations coming from $\tr(M_1^N)=0$ and $\tr(M_2^N)=0$ do not share any variables.
\end{proof}
\begin{problem}
	What can be said about the higher homotopy and homology of $\IMPS(D,d,N)$ and its projectivisation?
\end{problem}
\subsection{Surjectivity}\label{section:surjectivity}
In this section we investigate for which parameters the set $\IMPS(D,d,N)$ fills the ambient space $\Cyc^N(\mathbb{C}^d)$.
In the first part we prove that, for fixed $d$ and~$N$, $\IMPS(D,d,N)=\Cyc^N(\mathbb{C}^d)$ if $D$ is large enough. More precisely, it suffices to take $D \geq N \cdot \dim(\Cyc^N(\mathbb{C}^d))$ (see Corollary \ref{fillsforNlarge}).
In the second part of this section we study what happens if we drop the assumption $D \geq N$.
In the last part we describe very useful surjectivity criterion for polynomial maps, 
and apply it to the case $(D,d,N)=(3,2,4)$. 

The following result slightly generalizes \cite[Proposition 3.1]{GesmundoLandsbergWalter}.
\begin{prop} \label{linearSpan}
	If $D \geq N$, the linear span of $\IMPS(D,d,N)$ is equal to the whole space $\Cyc^N(\mathbb{C}^d)$.
\end{prop}
\begin{proof}
	The case $d \geq N$ was proven in \cite[Proposition 3.1]{GesmundoLandsbergWalter}.\\
	Suppose $d<N$. Then the projection $(\mathbb{C}^{N})^{\otimes N} \to (\mathbb{C}^d)^{\otimes N}$  maps $\IMPS(D,N,N)$ to $\IMPS(D,d,N)$ as in Lemma \ref{lemma:proj}. The proposition follows.
\end{proof}
 Now we will analyse how Proposition \ref{linearSpan} fails if we drop the assumption $D \geq N$. 
We start with a trivial example:
\begin{example}
	If $D=1$ then $\IMPS(1,d,N) \subseteq\mbox{Sym}^N(\mathbb{C}^d)$, which is a strict subspace of $\Cyc^N(\mathbb{C}^d)$ unless $N \leq 2$, $d \leq 1$, or $(d,N)=(2,3)$.
\end{example}
More surprisingly, even for $D>1$ it might still happen that $\IMPS(D,d,N)$ is contained in a strict linear subspace of $\Cyc^N(\mathbb{C}^d)$, as the following example shows:
\begin{example}
	The set $\IMPS(2,2,6)$ is contained in a strict  linear subspace of $\Cyc^6(\mathbb{C}^2)$: one can check that for every pair of $2 \times 2$-matrices $M_0$ and $M_1$, it holds that $\tr(M_1^2M_0^2M_1M_0)=\tr(M_1^2M_0M_1M_0^2)$.
\end{example}

In the following theorem, let $C(N_0,N_1)$ denote the number of sequences consisting of $N_0$ times the symbol '$0$' and $N_1$ times the symbol '$1$', where we identify two sequences if they are the same up to cyclic permutation.
\begin{theorem}
If
\[
C(N_0,N_1) > \binom{N_0+D-1}{D-1}\binom{N_1+D^2-D}{D^2-D}
\]
then for every $d \geq 2$, $\overline{\IMPS(D,d,N_0+N_1)}$ is contained in a strict linear subspace of $\Cyc^{N_0+N_1}(\mathbb{C}^d)$.
\end{theorem}
\begin{proof}
	It clearly suffices to show the theorem for d=2. 
	As in the proof of Theorem \ref{expdim}, $\overline{\IMPS(D,d,N)}$ is the closure of the image of a polynomial map $\phi: Z_2 \to \Cyc^N(\mathbb{C}^2)$, where $Z_2$ is an $D^2+1$-dimensional affine space. 
	Consider the set of all polynomials $\tr(A_{i_1}\cdots A_{i_N})$, where exactly $N_0$ of the indices are $0$ and the other $N_1$ indices are $1$. These polynomials have degree $N_0$ in the first $D$ variables, and degree $N_1$ in the last $D^2-D+1$ variables. The space of such polynomials has dimension $\binom{N_0+D-1}{D-1}\binom{N_1+D^2-D}{D^2-D}$. Hence, by the assumption, some of these polynomials must be linearly dependent. This imposes a linear condition on the image of $\phi$.
\end{proof}

\begin{remark}
For large $N_0,N_1$ the assumptions of previous theorem happen very often, as the left hand side grows exponentially, while the right hand side polynomially. 
\end{remark}

Our next goal is to show that $\IMPS(D,d,N)$ fills the ambient space for large~$D$.
For $X_1,\ldots,X_m$ subsets of a vector space $V$, $\jn(X_1,\ldots,X_m)$ will denote their join $\bigcup_{x_i \in X_i}\Span( x_1,\ldots, x_m)$, where $\Span$ is the affine span. 
\begin{lemma} \label{joinlemma}
	Let $X_i=\IMPS(D_i,d,N)$ for $1\leq i \leq m$. Then
	\[
	\jn(X_1,\ldots,X_m) \subseteq \IMPS(\sum D_i,d,N) \text{.}
	\]
\end{lemma}
\begin{proof}
	Let $v \in \jn(X_1,\ldots,X_m)$. Then $v = \sum_{i=1}^{m}{v_i}$, where $v_i \in \IMPS(D_i,d,N)$. There exist $D_i \times D_i$-matrices $M_{i,j}$ such that $v_j=T_N(M_{i,0},\ldots,M_{i,{d-1}})$. Thus
	\begin{align*}
	v &= \sum_{j=1}^{m}{T_N(M_{j,0},\ldots,M_{j,{d-1}})} \\
	&= T_N\Big(\diag(M_{1,0},\ldots,M_{m,0}),\ldots,\diag(M_{1,{d-1}},\ldots,M_{m,{d-1}})\Big)\\
	&\in \IMPS(\sum D_i,d,N) \text{.}
	\end{align*}
\end{proof}
\begin{cor} \label{secant}
	Suppose the linear span of $\IMPS(D,d,N)$ equals $\Cyc^N(\mathbb{C}^d)$. Then $\IMPS(D \cdot  \dim(\Cyc^N(\mathbb{C}^d)),d,N)=\Cyc^N(\mathbb{C}^d)$.
\end{cor}
\begin{proof}
	Writing $m=\dim(\Cyc^N(\mathbb{C}^d))$, by assumption $\Cyc^N(\mathbb{C}^d)$ is equal to the join of $m$ copies of $\IMPS(D,d,N)$. The result now follows from Lemma \ref{joinlemma}.
\end{proof}
\begin{cor} \label{fillsforNlarge}
	$\IMPS(D,d,N) = \Cyc^N(\mathbb{C}^d)$ for $D \geq N \cdot \dim(\Cyc^N(\mathbb{C}^d))$. 
\end{cor}
\begin{proof}
	Follows immediately from Propostition \ref{linearSpan} and Corollary \ref{secant}.
\end{proof}

\vspace*{1ex}
\begin{remark}
	Checking whether $\overline{\IMPS(D,d,N)}$ fills the ambient space is a computationally easy task for small parameter values. Indeed, since the set $\IMPS(D,d,N)$ is a constructible subset, its closure fills the ambient space if and only if it is full-dimensional. In particular for all cases in Table \ref{thetable}, $\overline{\IMPS(D,d,N)}$ fills the ambient space if and only if $(d-1)D^2 + 1 \geq \dim(\Cyc^N(\mathbb{C}^d))$ (see also Remark \ref{rmk:dimJac})).
	However, checking whether $\IMPS(D,d,N)$ fills the ambient space for fixed parameter values is a significantly more difficult task, which we address now.
\end{remark}

In this part we describe a general criterion for $\IMPS(D,d,N)$ to fill the space, and apply it to show that $\IMPS(3,2,4)$ fills the space.\\
Let $f\colon\mathbb{C}^n\rightarrow\mathbb{C}^m$ be a polynomial map
defined by $x\mapsto(f_1(x),\ldots,f_m(x))$, where the $f_i$ are homogeneous polynomials of the same degree in the coordinates of $x=(x_1,\ldots,x_n)$.
Instead of focusing on the above map
$f$, we want to consider the rational projective map
$\mathbb{P}^{n-1}\dasharrow\mathbb{P}^{m-1}$.
We prove the following theorem.

\begin{theorem}\label{baselocandsurj}
	Let $f\colon\mathbb{P}^{n-1}\dasharrow\mathbb{P}^{m-1}$ be a rational projective map 
	and $B$ be the base locus of $f$.
	If $$\dim(B)+\dim(Im(f))<n-1,$$ then the map $f$ has closed image. 
\end{theorem}
\begin{proof}
Let $d$ be the dimension of the image $Im(f)$ and $b$ the dimension
of the base locus $B$. Then all non-empty fibers of $f$ have at least dimension $n-1-d$.
Now we take the restriction to a generic  subspace $\mathbb{P}^d\subset\mathbb{P}^{n-1}$.
As the subspace is generic, we have that
$$\overline{Im}(f)=\overline{Im}(f|_{\mathbb{P}^d}).$$
Indeed, 
the subspace $\mathbb{P}^d$ intersects every fiber and we may assume $\mathbb{P}^d\cap B=\emptyset$, as $b+d<n-1$.
Therefore $f|_{\mathbb{P}^d}$ must have the same image as $f$.
 However, $f$ is well defined on $\mathbb{P}^d$, which is compact, and hence the image is closed.
\end{proof}

Using the above theorem we can deduce that for showing the surjectivity of~$f$ 
it is enough to find a sufficiently big space
on which the map $f$ is well defined.
\newline In the above notation we have the following corollary:
\begin{cor}\label{corbaselocandsurj}
	If there exists a subspace $Y\subset\mathbb{P}^{n-1}$ of dimension $Im(f)$, which is disjoint
	with the base locus $B$, then the map $f$ has closed image.
\end{cor}
Below we provide an application of the above corollary.
\begin{example} \label{eg:324fills}
	Let us consider the case $(D,d,N)=(3,2,4)$.
	Using Theorem~\ref{baselocandsurj} and Corollary \ref{corbaselocandsurj},
	to show that $f$ is surjective it is enough to find 
	a~$\mathbb{P}^5\subset\mathbb{P}^9$, disjoint with the base locus $B$.
	Below we present the \textit{Macaulay2} code  providing the $\mathbb{P}^5$, which satisfies the above assumptions.\\
	 \\
	\texttt{
	R=QQ[a_1,b_1,c_1,c_2,e_2,h_2];\\
	M1=matrix\{\{a_1,b_1,c_1\},\{c_2,e_2,h_2\},\\
	\{c_1+b_1-3*c_2,h_2-e_2,2*a_1-7*b_1\}\};\\
	M2=matrix\{\{a_1+2*c_2,e_2+5*h_2,-c_1-3*e_2\},\\
	\{b_1+a_1-2*h_2,e_2-5*c_2,b_1-c_1+13*a_1\},\\
	\{h_2-b_1-c_1,c_2+3*a_1-2*e_2,a_1-h_2\}\};\\
	a=trace(M1*M1*M1*M1);\\
	b=trace(M1*M1*M1*M2);\\
	c=trace(M1*M1*M2*M2);\\
	d=trace(M1*M2*M1*M2);\\
	e=trace(M1*M2*M2*M2);\\
	f=trace(M2*M2*M2*M2);\\
	I=ideal(a,b,c,d,e,f);\\
	(dim I) == 0}
\end{example}

\begin{remark}
We now explain how to find the given $\PP^5$. Taking a completely general $\PP^5$, although from a theoretical point of view most desirable, is not possible due to computational restraints. On the other hand taking very special, simple $\PP^5$ usually leads to intersection with $B$ that is of large dimension. The given $\PP^5$ was found by first considering a special $\PP^5$ and computing the dimension and degree of the intersection. The $\PP^5$ was successively modified to a more general one, each time computing the dimension and degree. The degree (in most cases) or dimension of the intersection were dropping, while we modified the $\PP^5$. This meant that we were not in a generic situation and further modifications were possible. Finally, we reached the given example.  
\end{remark}

\subsection{The trace parametrization}\label{section:tracepar}
We start this section by recalling briefly another parametrization of the matrix product states.
Let us consider $d$ matrices $M_0,\ldots, M_{d-1}$ in $\mathbb{K}^{D\times D}$ with indeterminate entries. The \emph{trace algebra}  ${\cal C}_{D,d}$ is the algebra generated by the traces of products $\tr(M_{i_1}\cdots M_{i_k})$.
 It can alternatively be described as the algebra of all polynomial expressions in the entries of the $M_i$ that are invariant under simultaneous conjugation \cite{Sibirskii, Leron, Procesi}.
 
It follows from a standard fact in invariant theory that the trace algebra is generated by finitely many traces $T_1, \ldots, T_K$, where every $T_j$ is an expression of the form $\tr(M_{i_1}\cdots M_{i_k})$. It follows that ${\cal C}_{D,d}$ is isomorphic to a quotient algebra $\CC[T_1,\ldots, T_K]/(f_1,\ldots, f_r)$ for some polynomials $f_i(T_1,\ldots,T_K)$.
We will now parametrize $\IMPS(D,d,N)$ with the \emph{spectrum} of the algebra ${\cal C}_{D,d}$. For readers not familiar with the $\Spec$ construction: $\Spec({\cal C}_{D,d})$ may be regarded as the set of all $K$-tuples $(t_1,\ldots,t_K)$ satisfying the equations $f_i(t_1,\ldots,t_K)=0$.
We have the following diagram 
\begin{center}
\begin{tikzcd}
(\CC^{D\times D})^d \arrow[rr,"T_N"] \arrow[dr, twoheadrightarrow, "\pi"] & & \Cyc^N(\CC^d) \\
& \Spec {\cal C}_{D,d} \arrow[ur,"\widetilde{T_N}"] &  
\end{tikzcd}
\end{center}
where the surjectivity of the left map is again a standard fact in invariant theory. Hence, the maps $T_N$ and $\widetilde{T_N}$ have the same image $\IMPS(D,d,N)$.

Sibirskii \cite{Sibirskii} showed in the case $D=2$ that the trace algebra ${\cal C}_{2,d}$
is minimally generated by the elements $\mbox{tr}(M_i)$, $\mbox{tr}(M_iM_j)$ for $0\leq i \leq j\leq d-1$, and $\mbox{tr}(M_iM_jM_k)$
for $0\leq i<j<k\leq d-1$.
Moreover in the case $D=d=2$, there are no relations between the five generators $\tr(M_0)$, $\tr(M_1)$, $\tr(M_0^2)$, $\tr(M_0M_1)$, $\tr(M_1^2)$. In other words, ${\cal C}_{2,2}$ is the polynomial ring in 5 variables. This means that we get a parametrization $\widetilde{T_N}\colon\mathbb{C}^5\rightarrow\IMPS(2,2,N)$.\\

Using the trace parametrization and \textit{Macaulay2}, it is possible to obtain equations for $\IMPS(2,2,N)$ for small values of $N$.

\begin{theorem}
	\begin{enumerate}
		\item{\cite[Theorem 3]{CritchMorton}} The ideal of $\IMPS(2,2,4) \subset \CC^6$ is generated by one sextic.
		\item{\cite[Question after Theorem 4]{CritchMorton}} The ideal of $\IMPS(2,2,5) \subset \CC^8$ is generated by 3 quadrics and 27 sextics. 
		\item The ideal of $\IMPS(2,2,6) \subset \CC^{14}$ is generated by 1 linear form, 6 quadrics, and 17 cubics.
	\end{enumerate}
\end{theorem}
The equations, as well as the code we used to obtain them, can be found online at \cite{webpageTim}. 

The following conjecture is closely related to \cite[Conjecture 11.9]{BrayMortonHMM} in algebraic statistics:
\begin{conjecture}
	For any fixed $D$ and $d$, there is an $M$ such that the ideal of $\IMPS(D,d,N)$ is generated by quadrics for all $N \geq M$.
\end{conjecture}

Using the trace parametrization, and the invariance of $\IMPS(D,d,N)$ under local transformations, we were able to obtain a complete description of the set $\IMPS(2,2,4)\subseteq \Cyc^4(\CC^2)$: it can be obtained by removing three $GL_2$-orbits from a degree six hypersurface in $\Cyc^4(\CC^2) \cong \CC^6$. For more details, see Appendix \ref{appendix:224}.

\subsection{The fundamental theorem}\label{section:fundtheorem}
In the literature appear several versions of the fundamental theorem of matrix product states, which all roughly say that for $N$ large enough the map parametrizing matrix product states is generically injective up to obvious symmetry.
The following formulation is adapted from \cite[Corollary 7]{Molnar}.
\begin{theorem} \label{thm:fundthm}
	Let $A_0,\ldots,A_{d-1}\in\mathbb{C}^{D\times D}$ and $B_0,\ldots,B_{d-1}\in\mathbb{C}^{D\times D}$ be such that $T_N(A_0,\ldots,A_{d-1})=T_N(B_0,\ldots,B_{d-1})$ and assume that $N\geq 2L+1$, where $L$ is such that $\Span(\{A_0,\ldots,A_{d-1}\}^L)=\Span(\{B_0,\ldots,B_{d-1}\}^L)=\CC^{D\times D}$.
	Then there is an invertible matrix $Z$ and a constant $\zeta\in\CC$ with $\zeta^N=1$, such that $B_i=\zeta Z^{-1} A_iZ$ for every $i$. Moreover $Z$ is unique up to a multiplicative constant.
\end{theorem}
We recall the generic injectivity index $GC_{D,d}$ from Remark \ref{rmk:geninin}: it is the lowest number such that the following holds: for a generic tuple $\{A_0,\ldots,A_{d-1}\}$ of $D\times D$-matrices $\Span(\{A_0,\ldots,A_{d-1}\}^{GC_{D,d}})=\CC^{D\times D}$.
\begin{cor} \label{cor:fundthm}
	Assume that $N\geq 2GC_{D,d}+1$, then for a generic tensor in $\IMPS(D,d,N)$
	all preimages under the map $T_N$ are the same up to simultaneous conjugation and multiplication 
	by an $N-$th root of unity.
\end{cor}
This would be a trivial corollary of Theorem \ref{thm:fundthm} if we could restrict the map $T_N$ to the dense open subset of tuples $\mathcal{A} \in (\CC^{D\times D})^d$ satisfying $\Span(\mathcal{A}^L)=\CC^{D\times D}$. However, it is a priori not clear that the complement of this set cannot map to a dense subset of $\IMPS(D,d,N)$. In the rest of this section, we show that this indeed does not happen, using the trace paramatrization.

Before we start the proof, we introduce the following notation:
A point $\mathcal{A}=(A_0,\dots,A_{d-1}) \in (\CC^{D \times D})^d$ gives rise to a $D$-dimensional representation $\varphi_{\mathcal{A}}$ of the associative algebra $\CC \langle X_0, \dots, X_{d-1} \rangle$, by putting $\varphi_{\mathcal{A}}(X_i)=A_i$
\begin{theorem}[{See \cite[II.2.7]{kraft}}]
	\begin{itemize}
		\item The association $\mathcal{A} \mapsto \varphi_{\mathcal{A}}$ induces a bijection between the $GL_D$-orbits of $(\CC^{D \times D})^d$, and the set of $D$-dimensional representations of $\CC \langle X_0, \dots, X_{d-1} \rangle$ up to isomorphism.
		\item The orbit of $\mathcal{A}$ is closed if and only if $\varphi_{\mathcal{A}}$ is a semisimple representation.
	\end{itemize}
\end{theorem}
\begin{cor} \label{cor:almostgeo}
	The orbit of a general point in $(\CC^{D \times D})^d$ is closed. In other words: $(\CC^{D \times D})^d \to \Spec{\cal C}_{D,d}$ is an almost geometric quotient. 
\end{cor}
\begin{proof}
	We claim that for a general point $\mathcal{A}=(A_0,\dots,A_{d-1}) \in (\CC^{D \times D})^d$, the representation $\varphi_{\mathcal{A}}$ is simple. 
	Indeed: $\varphi_{\mathcal{A}}$ is simple if and only if there is no nontrivial subspace of $\CC^D$ simultaneously fixed by all $A_i$, which is clearly the case for a generic choice of $\varphi_{\mathcal{A}}$. (For example: if all $A_i$ are diagonalizable, having a subspace fixed by $A_0$ and $A_1$ would in particular imply that there is an eigenvector of $A_1$ that is a linear combination of less than $D$ eigenvectors of~$A_0$; a nongeneric condition.)
\end{proof}
\begin{proof}[Proof of Corollary \ref{cor:fundthm}]
	Since $\pi:(\CC^{D \times D})^d \to \Spec{\cal C}_{D,d}$ is an almost geometric quotient (Corollary \ref{cor:almostgeo}), there is a dense open subset $U \subseteq \Spec{\cal C}_{D,d}$ such that $\pi|_{\pi^{-1}(U)}$ induces a bijection between $GL_D$-orbits in $\pi^{-1}(U)$ and points in $U$.
	
	We show that $\Spec{\cal C}_{D,d}$ and $\IMPS(D,d,N)$ both have dimension $(d-1)D^2+1$. For $\Spec{\cal C}_{D,d}$ the dimension can be computed as the difference $\dim((\CC^{D \times D})^d)-\dim(\pi^{-1}(x))$, where $x$ is a general point in $\Spec{\cal C}_{D,d}$. We can assume $x \in U$, so that $\pi^{-1}(x)$ is the $GL_D$-orbit of a generic $d$-tuple of matrices, which clearly has dimension $D^2-1$. For $\IMPS(D,d,N)$: let $V \subseteq (\CC^{D \times D})^d$ be set of tuples $\mathcal{A}$ for which $\mathcal{A}^{GC_{D,d}}=\CC^{D\times D}$. By the definition of $GC_{D,d}$, $V$ is dense, hence $\dim \IMPS(D,d,N) = \dim(T_N(V))$. Now we can do the same computation as before: by Theorem \ref{thm:fundthm}, a general fiber of $T_N|_V$ has dimension $D^2-1$, and we conclude $\dim(\IMPS(D,d,N))=(d-1)D^2+1$.\\
	We claim that $\pi((\CC^{D \times D})^d \setminus V)$ is contained in a lower-dimensional subspace of $\Spec{\cal C}_{D,d}$. Indeed, consider $Y=\pi((\CC^{D \times D})^d \setminus V) \cap U$. Then since $V$ is $GL_D$-invariant and by definition of $U$, it holds that $\pi^{-1}(Y)=((\CC^{D \times D})^d \setminus V) \cap \pi^{-1}(U)$. So $\pi^{-1}(Y)$ is not Zariski dense, hence the same holds for $\pi((\CC^{D \times D})^d \setminus V)$.
	
	Now consider the map $\widetilde{T_N}:\Spec{\cal C}_{D,d} \twoheadrightarrow \IMPS(D,d,N)$. Since both spaces have the same dimension, the lower-dimensional set $\Spec{\cal C}_{D,d} \setminus V$ will map to a lower-dimensional subset of $\IMPS(D,d,N)$. So for a general point $x \in \IMPS(D,d,N)$, we get $\widetilde{T_N}^{-1}(x) \subseteq V$. Then we are done by Theorem \ref{thm:fundthm}.
\end{proof}

The above Corollary can also be stated as follows: assume that $N\geq 2GC_{D,d}+1$, then for a generic tensor in $\IMPS(D,d,N)$
there are exactly $N$ preimages under $\widetilde{T_N}$. 
In the case $D=d=2$, one easily checks that $GC_{D,d}=2$, yielding the following result, which was stated as a conjecture in \cite[Conjecture 12]{CritchMorton}.
\begin{theorem} \label{thm:85678}
	Using the trace parametrization $\varphi_N$, for $N\geq 5$ almost every matrix product state in $\IMPS(2,2,N)$ has exactly $N$ choices 
	of parameters that yield it.
\end{theorem}

\newpage
\appendix
\section{Proof of Theorem \ref{ClosedComplex}}
The following lemma concludes the first part of proof of Theorem \ref{ClosedComplex}.
\begin{lemma} \label{appendix233}
	$T \notin \IMPS(2,3,3)$. 
\end{lemma}
\begin{proof}
    We assume that $T \in \IMPS(2,3,3)$ and derive a contradiction. Then, we may write 
	\[
	T=e_{012} = T_3(M_{0},M_{1},M_{2})
	\]
	for some $M_0,M_1,M_2 \in \Mat(2 \times 2, \mathbb{C})$, where $M_0$ is in Jordan normal form.\\
	Suppose first that $M_0 = 
	\begin{pmatrix}
	a & 1 \\
	0 & a
	\end{pmatrix}
	$ . 
	Since $\tr(M_0^3)=0$, we get $a=0$. Let us furthermore write  
	$M_1 = 
	\begin{pmatrix}
	a_2 & b_2 \\
	c_2 & d_2
	\end{pmatrix}
	$ and 
	$M_2 = 
	\begin{pmatrix}
	a_3 & b_3 \\
	c_3 & d_3
	\end{pmatrix}
	$.\\ Then since $\tr(M_0M_1^2)=\tr(M_0M_2^2)=\tr(M_0M_2M_1)=0$ and $\tr(M_0M_1M_2) = 1$, we get that $(a_2,b_2,c_2,d_2,a_3,b_3,c_3,d_3)$ must be a solution of the following system of four equations:
	\vspace*{1.5ex}
	
	$\begin{cases}
	c_2(a_2+d_2)=0 \\
	c_3(a_3+d_3)=0 \\
	c_3a_2+c_2d_3=0\\
	c_2a_3+c_3d_2=1\\
	\end{cases}$
	\vspace*{1.5ex}
	
	It is not hard to see that this system has no solutions.\\
	Suppose now that $M_0 = 
	\begin{pmatrix}
	a_1 & 0 \\
	0 & d_1
	\end{pmatrix}
	$ , and write $M_1 = 
	\begin{pmatrix}
	a_2 & b_2 \\
	c_2 & d_2
	\end{pmatrix}
	$ and 
	$M_2 = 
	\begin{pmatrix}
	a_3 & b_3 \\
	c_3 & d_3
	\end{pmatrix}$ as before. Then $(a_1,d_1,a_2,b_2,c_2,d_2,a_3,b_3,c_3,d_3)$ must be a solution of the following system:
	\vspace*{1.5ex}
	
	$\begin{cases}
	a_1^3 + d_1^3 = 0\\
	a_2^3 + 3a_2b_2c_2 + 3b_2c_2d_2 + d_2^3 =0\\
	a_3^3 + 3a_3b_3c_3 + 3b_3c_3d_3 + d_3^3=0\\
	a_1a_2^2 + a_1b_2c_2 + b_2c_2d_1 + d_1d_2^2=0\\
	a_1^2a_2 + d_1^2d_2=0\\
	a_1a_3^2 + a_1b_3c_3 + b_3c_3d_1 + d_1d_3^2=0\\
	a_1^2a_3 + d_1^2d_3=0\\
	a_2a_3^2 + a_3b_3c_2 + a_3b_2c_3 + a_2b_3c_3 + b_3c_3d_2 + b_3c_2d_3 + b_2c_3d_3 + d_2d_3^2=0\\
	a_2^2a_3 + a_3b_2c_2 + a_2b_3c_2 + a_2b_2c_3 + b_3c_2d_2 + b_2c_3d_2 + b_2c_2d_3 + d_2^2d_3=0\\
	a_1a_2a_3 + a_1b_2c_3 + b_3c_2d_1 + d_1d_2d_3 = 1\\
	a_1a_2a_3 + a_1b_3c_2 + b_2c_3d_1 + d_1d_2d_3 =0
	\end{cases}$.
	\vspace*{1.5ex}
	
	One can show that this system has no solutions for example by computing a Gr\"obner basis in \textit{Macaulay2}. This leads to a contradiction.
\end{proof}
The next lemma concludes the second part of the proof of Theorem \ref{ClosedComplex}.
\begin{lemma} \label{appendix22N}
	For $N>3$, $e_{0\cdots01} \notin \IMPS(2,2,N)$. 
\end{lemma}
\begin{proof}
Let us write $T=e_{0\cdots01}$ and assume that  $T\in \IMPS(2,2,N)$. Then we can write $T=T_N(M_0,M_1)$, for some $2 \times 2$-matrices $M_0$ and $M_1$, where $M_0$ is in Jordan normal form.\\
First suppose $M_0 = 
\begin{pmatrix}
a & 1 \\
0 & a
\end{pmatrix}
$ . 
Then since $\tr(M_0^N)=0$, we get $a=0$. But then also $\tr(M_0^{N-1}M_1)=0$. This is a contradiction with our assumption that $\tr(M_0^{N-1}M_1)=1$. \\
Thus, we can assume that $M_0 = 
\begin{pmatrix}
a & 0 \\
0 & d
\end{pmatrix}
$ and
write $M_1 = 
\begin{pmatrix}
A & B \\
C & D
\end{pmatrix}
$.\\
First note that since $\tr(M_0^N)=0$, we get that $a=\zeta d$, where $\zeta^N = -1$. It is clear that $a \neq 0$ and $d \neq 0$, since otherwise $M_0=0$. \\
The equation $\tr(M_0^{N-1}M_1)=1$ becomes
\begin{align}
& a^{N-1}A+d^{N-1}D=1 \nonumber \\
\iff & a^{N-1}A+a^{N-1}\zeta^{N-1}D=1 \nonumber \\
\implies & A+\zeta^{N-1}D \neq 0 \text{.} \label{tr1}
\end{align}
We also get that, for every $s \in \{0,1,\ldots N-2 \}$:
\begin{align}
& \tr(M_0^sM_1M_0^{N-2-s}M_1)=0 \nonumber \\ 
\iff & a^{N-2}A^2 + (a^sd^{N-2-s}+a^{N-2-s}d^s)BC + d^{N-2}D^2 = 0 \nonumber \\
\iff & A^2 + (\zeta^{s}+\zeta^{N-2-s})BC+\zeta^{N-2}D^2 = 0 \text{.} \label{tr2}
\end{align} 
Furthermore, for every $t \in \{0,1,\ldots N-3 \}$: 
\begin{align}
& \tr(M_0^tM_1^2M_0^{N-3-t}M_1) = 0 \nonumber \\
\iff & a^{N-3}A(A^2+BC) + (a^td^{N-3-t}+a^{N-3-t}d^t)BC(A+D)+d^{N-3}D(D^2+BC) = 0 \nonumber \\
\iff & A(A^2+BC)+(\zeta^t+\zeta^{N-3-t})BC(A+D)+\zeta^{N-3}D(D^2+BC) = 0 \text{.} \label{tr3}
\end{align}
We now show that the equalities (\ref{tr2}) and (\ref{tr3}), together with the inequality (\ref{tr1}), lead to a contradiction. The proof is not hard, but we need to distinguish some cases. In the proof it will turn out that we only need (\ref{tr2}) for $s \in \{0,1,2\}$, and (\ref{tr3}) for $t \in \{0,1\}$.
\begin{case}
	$BC = 0$.
\end{case}
Then we get that
\begin{align*}
A+\zeta^{N-1}D \neq 0 \\
A^2+\zeta^{N-2}D^2 = 0 \\
A^3+\zeta^{N-3}D^3 = 0
\end{align*}
So
\begin{align*}
& \zeta^{N-2}AD^2 = - A^3 = \zeta^{N-3}D^3 \\
\implies & D^2(\zeta A - D) = 0 \text{,}
\end{align*}
but this leads to a contradiction: either $\zeta A - D = 0$, but this is a contradiction with (\ref{tr1}); or $D=0$, which implies $A=0$ hence also yields a contradiction with~(\ref{tr1}). 

\begin{case}
	$BC \neq 0$.
\end{case}
Then (\ref{tr2}) tells us that for every $s \in \{0,1,\ldots N-2 \}$:
\begin{align*}
& 1 + \zeta^{N-2} = \zeta^s + \zeta^{N-2-s} \\
\implies & (1-\zeta^s)(1-\zeta^{N-2-s}) = 0 \text{.}
\end{align*}
Putting $s=1$ yields $\zeta^{N-3}=1$.\\
Putting $s=2$ yields $\zeta^{N-4}=1$ or $\zeta^2=1$. The former would imply $\zeta=1$, a contradiction. So we get $\zeta=-1$ (and $N$ odd). Note that here we used that $N \geq 4$. \\
Now (\ref{tr1}) tells us that $A+D \neq 0$. 
But then by (\ref{tr3}) for $t=0$ and $t=1$, we get $1+\zeta^{N-3}=\zeta+\zeta^{N-2}$. Since $\zeta=-1$ and $N$ is odd, that is a contradiction.
\end{proof}

\section{$\IMPSbis(2,2,4)$ as a constructible set}\label{appendix:224}

In this section we give a description of $\IMPS(2,2,4)$ as a constructible subset of $\Cyc^4(\CC^2) \cong \CC^6$. See \cite{webpageTim} for \textit{Macaulay2} code accompanying this section.\\
 Using the notation from Section \ref{section:def}, $\Cyc^4(\CC^2)$ has a basis given by $e_{0000}$, $e_{0001}$, $e_{0011}$, $e_{0111}$, $e_{1111}$, $e_{0101}$. We write coordinates on the space $\Cyc^4(\CC^2)$ by $x_{0000},\ldots,x_{0101}$.\\
The closure $\overline{\IMPS(2,2,4)}$ was already computed in \cite[Theorem 3]{CritchMorton}: it is a hypersurface cut out by the polynomial
\begin{dmath*}
	f_{224}=2x_{0011}^6-12x_{0001}x_{0011}^4x_{0111}+16x_{0001}^2x_{0011}^2x_{0111}^2+4x_{0000}x_{0011}^3x_{0111}^2-8x_{0000}x_{0001}x_{0011}x_{0111}^3+x_{0000}^2x_{0111}^4+4x_{0001}^2x_{0011}^3x_{1111}-x_{0000}x_{0011}^4x_{1111}-8x_{0001}^3x_{0011}x_{0111}x_{1111}+2x_{0000}x_{0001}^2x_{0111}^2x_{1111}+x_{0001}^4x_{1111}^2+8x_{0001}x_{0011}^3x_{0111}x_{0101}-16x_{0001}^2x_{0011}x_{0111}^2x_{0101}-4x_{0000}x_{0011}^2x_{0111}^2x_{0101}+4x_{0000}x_{0001}x_{0111}^3x_{0101}-4x_{0001}^2x_{0011}^2x_{1111}x_{0101}+4x_{0001}^3x_{0111}x_{1111}x_{0101}+8x_{0000}x_{0001}x_{0011}x_{0111}x_{1111}x_{0101}-2x_{0000}^2x_{0111}^2x_{1111}x_{0101}-2x_{0000}x_{0001}^2x_{1111}^2x_{0101}-x_{0011}^4x_{0101}^2+4x_{0001}^2x_{0111}^2x_{0101}^2+4x_{0000}x_{0011}x_{0111}^2x_{0101}^2+4x_{0001}^2x_{0011}x_{1111}x_{0101}^2-2x_{0000}x_{0011}^2x_{1111}x_{0101}^2-4x_{0000}x_{0001}x_{0111}x_{1111}x_{0101}^2+x_{0000}^2x_{1111}^2x_{0101}^2-2x_{0000}x_{0111}^2x_{0101}^3-2x_{0001}^2x_{1111}x_{0101}^3+x_{0000}x_{1111}x_{0101}^4
\end{dmath*}
as can be verified by a Gr\"obner basis computation, for example in \emph{Macaulay2}.\\
In principle, we could use \emph{TotalImage.m2} (see \cite{HarrisMichalekSertoz})  to compute the image of the trace parametrization map $\widetilde{T}:\CC^5 \to \Cyc^4(\CC^2)$. This computation did not finish in a reasonable amount of time. However, as we will explain now, one can exploit symmetries of $\IMPS(2,2,4)$ to simplify the computations.\\
Recall from Remark \ref{rmk:locsym} that $\IMPS(2,2,4)$ is invariant under the natural $GL_2$-action on $\Cyc^4(\CC^2)$. 
We use the following strategy
\begin{enumerate}
	\item Find a low-dimensional subset $Y \subseteq \Cyc^4(\CC^2)$ that contains at least one point from every $GL_2$-orbit.
	\item Use \emph{TotalImage.m2} to compute $Z=\widetilde{T}(\widetilde{T}^{-1}(Y))$.
	\item Compute $GL_2 \cdot (Y \cap Z)$. 
\end{enumerate}
We now describe this in more detail. First, note that
$$\Cyc^4(\CC^2)\cong \Sym^4(\CC^2)\oplus \CC,$$ where the map from $\CC$ to $\Cyc^4(\CC^2)$ is given by sending $1$ to $e_{0011}-2e_{0101}$ and
the map from  $\Sym^4(\CC^2)$ to $\Cyc^4(\CC^2)$ is given by
\begin{align*}
x^4 & \mapsto e_{0000} \\
x^3y & \mapsto \frac{1}{4}e_{0001} \\
x^2y^2 & \mapsto \frac{1}{6}(e_{0011}+e_{0101}) \\
xy^3 & \mapsto \frac{1}{4}e_{0111} \\
y^4 & \mapsto e_{1111}.
\end{align*}
Because $\Sym^4(\CC^2)$ can be seen as the space of homogeneous degree $4$ polynomials in $2$ variables one can easily see that the following set 
contains exactly one representative of every $GL_2-$orbit 
$$\{x^4,x^3y,x^2y^2\}\cup\{xy(x-y)(x-\mu y)\,|\,\mu\in\CC\setminus \{1\}\}.$$  
We can deduce that 
if we define the following subsets of $\Cyc^4(\CC^2)$
\begin{align*}
	Y_1 &= V(x_{0001},x_{0111},x_{1111},2x_{0011}+x_{0101})\\
	Y_2 &= V(x_{0000},x_{0111},x_{1111},2x_{0011}+x_{0101})\\
	Y_3 &= V(x_{0000},x_{0001},x_{0111},x_{1111})\\
	Y_4 &= V(x_{0000},x_{1111},2x_{0001}+2x_{0011}+2x_{0111}+x_{0101})
\end{align*}
then $Y=Y_1 \cup Y_2 \cup Y_3 \cup Y_4$ contains at least one point from every $GL_2$-orbit.
Using $\textit{TotalImage.m2}$, we computed $Z_i:=\widetilde{T}(\widetilde{T}^{-1}(Y_i))$ and compared it with $V(f_{224}) \cap Y_i$:
\begin{align*}
&Z_1 = \{e_{0000}\} = V(f_{224}) \cap Y_1 \\
&Z_2 = \emptyset  \neq \{e_{0001}\}=V(f_{224}) \cap Y_2 \\
&Z_3 = \{e_{0101}\}  \neq \{e_{0101}, e_{0011}+\sqrt{2}e_{0101}, e_{0011}-\sqrt{2}e_{0101}\}=V(f_{224}) \cap Y_3\\
&Z_4 = V(f_{224}) \cap Y_4.
\end{align*}
We conclude the following: $\IMPS(2,2,4)$ is the vanishing locus of the polynomial $f$, with the orbits of the following three tensors removed: $e_{0001}$, $e_{0011}+\sqrt{2}e_{0101}$,  $e_{0011}-\sqrt{2}e_{0101}$.
One can easily compute the orbit closures: they are cut out by the following ideals:
\begin{dmath*}
	I_1\hiderel{=}(x_{0101}-x_{0011},
	3x_{0011}^2-4x_{0001}x_{0111}+x_{0000}x_{1111},\\
	2x_{0001}x_{0011}x_{0111}-3x_{0000}x_{0111}^2-3x_{0001}^2x_{1111}+4x_{0000}x_{0011}x_{1111},\\
	8x_{0001}^2x_{0111}^2-9x_{0000}x_{0011}x_{0111}^2-9x_{0001}^2x_{0011}x_{1111}+14x_{0000}x_{0001}x_{0111}x_{1111}-4x_{0000}^2x_{1111}^2)
\end{dmath*}
\begin{dmath*}
	I_2 = (2x_{0001}x_{0111}+(2+2\sqrt{2})x_{0011}^2-(8+5\sqrt{2})x_{0011}x_{0101}+(4+3\sqrt{2})x_{0101}^2,\\
	(1-\sqrt{2})x_{1111}x_{0011}-x_{1111}x_{0101}+\sqrt{2}x_{0111}^2,\\
	x_{1111}x_{0001}-(2+\sqrt{2})x_{0111}x_{0011}+(1+\sqrt{2})x_{0111}x_{0101},\\ x_{0000}x_{0011}-(1+\sqrt{2})x_{0000}x_{0101}+\sqrt{2}x_{0001}^2,\\
	x_{0000}x_{0111}-(2+\sqrt{2})x_{0001}x_{0011}+(1+\sqrt{2})x_{0001}x_{0101},\\
	x_{0000}x_{1111}-(6+4\sqrt{2})x_{0011}^2+(8+6\sqrt{2})x_{0011}x_{0101}-(3+2\sqrt{2})x_{0101}^2)
\end{dmath*}
\begin{dmath*}
	I_3 = (2x_{0001}x_{0111}+(2-2\sqrt{2})x_{0011}^2-(8-5\sqrt{2})x_{0011}x_{0101}+(4-3\sqrt{2})x_{0101}^2,\\
	(1+\sqrt{2})x_{1111}x_{0011}-x_{1111}x_{0101}-\sqrt{2}x_{0111}^2,\\
	x_{1111}x_{0001}-(2-\sqrt{2})x_{0111}x_{0011}+(1-\sqrt{2})x_{0111}x_{0101},\\ x_{0000}x_{0011}-(1-\sqrt{2})x_{0000}x_{0101}-\sqrt{2}x_{0001}^2,\\
	x_{0000}x_{0111}-(2-\sqrt{2})x_{0001}x_{0011}+(1-\sqrt{2})x_{0001}x_{0101},\\
	x_{0000}x_{1111}-(6-4\sqrt{2})x_{0011}^2+(8-6\sqrt{2})x_{0011}x_{0101}-(3-2\sqrt{2})x_{0101}^2).
\end{dmath*}
Finally, one can check that $V(I_1)=(GL_2\cdot{e_{0001}}) \cup (GL_2\cdot{e_{0000}}) $, $V(I_2)=(GL_2\cdot(e_{0011}+\sqrt{2}e_{0101}))  \cup (GL_2\cdot{e_{0000}}) $, and $V(I_3)=(GL_2\cdot(e_{0011}-\sqrt{2}e_{0101}))  \cup (GL_2\cdot{e_{0000}}) $. Now $GL_2\cdot{e_{0000}} $, is a closed orbit consisting of all rank 1 symmetric tensors in $\Cyc^4(\CC^2)$. Explicitly, it is cut out by the ideal
\begin{dmath*}
	J=(x_{0101}-x_{0011},x_{0000}x_{0011}-x_{0001}^2,x_{0000}x_{0111}-x_{0001}x_{0011},x_{0000}x_{1111}-x_{0001}x_{0111},x_{0001}x_{0111}-x_{0011}^2,x_{0001}x_{1111}-x_{0011}x_{0111},x_{0011}x_{1111}-x_{0111}^2).
\end{dmath*}
Finally, we obtain the following description of $\IMPS(2,2,4) \subseteq \Cyc^4(\CC^2)$ as a constructible set:
\[
\IMPS(2,2,4) = (V(f) \setminus (V(I_1) \cup V(I_2) \cup V(I_3))) \cup V(J).
\]

\section*{Acknowledgement}

We would like to express our gratitude to W.~Hackbusch for inspiring discussions.
The first author would like to thank Micha\l{} Farnik for useful discussion about the proof of the Theorem~\ref{baselocandsurj}. MM was supported by the Polish National Science Center Project2013/08/A/ST1/00804 affiliated at the University of Warsaw.
The third author would like to thank Jutho Haegeman and Frank Verstraete for several useful discussions.

\bibliography{mybib}{}
\bibliographystyle{plain}

\subsubsection*{Authors’ addresses:}
Adam Czapli\'nski, Universit\"at Siegen, Department Mathematik, Walter-Flex-Stra{\ss}e 3, 57068 Siegen, Germany
\newline \url{adam.czaplinski@uni-siegen.de}
\newline Mateusz Micha\l{}ek,  Max Planck Institute for Mathematics in the Sciences, \linebreak Inselstra{\ss}e 22, 04103 Leipzig, Germany and Aalto University, Espoo, Finland
\newline \url{michalek@mis.mpg.de} 
\newline Tim Seynnaeve,  Max Planck Institute for Mathematics in the Sciences, Inselstra{\ss}e 22, 04103 Leipzig, Germany
\newline \url{tim.seynnaeve@mis.mpg.de}

\end{document}